\documentclass[a4paper,UKenglish,cleveref, autoref, thm-restate]{lipics-v2021}

\title{Sparsity-Parameterised Dynamic Edge Colouring} 

\author{Aleksander B. G. Christiansen}{Technical University of Denmark, Denmark}{abgch@dtu.dk}{https://orcid.org/0000-0002-9486-9115}{}

\author{Eva Rotenberg}{Technical University of Denmark, Denmark}{erot@dtu.dk}{https://orcid.org/0000-0001-5853-7909}{}

\author{Juliette Vlieghe}{Technical University of Denmark, Denmark}{jmvvl@dtu.dk}{https://orcid.org/0009-0004-0079-8523}{}

\authorrunning{A.B.G Christiansen, E. Rotenberg and J.M.V. Vlieghe} 

\Copyright{Aleksander Bjørn Grodt Christiansen, Eva Rotenberg, Juliette Vlieghe} 



\category{} 

\relatedversion{} 
\relatedversiondetails[linktext={}, cite=]{This paper is based on the master thesis of Juliette Vlieghe, supervised by Eva Rotenberg and Aleksander B. G. Christiansen and defended on 30th June 2023.}{http://dx.doi.org/10.13140/RG.2.2.18471.52648} 

\acknowledgements{We thank Jacob Holm for his interest in this work, and for comments and improvements to an earlier version of this manuscript.
We also thank an anonymous reviewer for their highly detailed feedback and suggestions for improvement.}

\nolinenumbers 

\funding{This work was supported by the VILLUM Foundation grant VIL37507 ``Efficient Recomputations for Changeful Problems''.}

\EventEditors{ }
\EventNoEds{0}
\EventLongTitle{}
\EventShortTitle{SWAT 2024}
\EventAcronym{SWAT}
\EventYear{2024}
\EventDate{ }
\EventLocation{ }
\EventLogo{}
\SeriesVolume{41}
\ArticleNo{X}

\usepackage{algorithm}
\usepackage{algorithmicx}
\usepackage[noend]{algpseudocode}

\usepackage{tikz}
\usetikzlibrary{arrows.meta,decorations.pathmorphing,backgrounds,positioning,fit,petri}
\usetikzlibrary{decorations.pathreplacing,calligraphy}

\begin{document}

\maketitle

\begin{abstract}

We study the edge-colouring problem, and give efficient algorithms where the number of colours is parameterised by the graph's arboricity, $\alpha$. 
In a dynamic graph, subject to insertions and deletions, we give a deterministic algorithm that updates a proper $\Delta + O(\alpha)$ edge~colouring in $\operatorname{poly}(\log n)$ amortized time. Our algorithm is fully adaptive to the current value of the maximum degree and arboricity.

In this fully-dynamic setting, the state-of-the-art edge-colouring algorithms are either a randomised algorithm using $(1 + \varepsilon)\Delta$ colours in $\operatorname{poly}(\log n, \epsilon^{-1})$ time per update, or the naive greedy algorithm which is a deterministic $2\Delta -1$ edge colouring with $\log(\Delta)$ update time. 

Compared to the $(1+\varepsilon)\Delta$ algorithm, our algorithm is deterministic and asymptotically faster, and when $\alpha$ is sufficiently small compared to $\Delta$, it even uses fewer colours. In particular, ours is the first $\Delta+O(1)$ edge-colouring algorithm for dynamic forests, and dynamic planar graphs, with polylogarithmic update time. 

Additionally, in the static setting, we show that we can find a proper edge colouring with $\Delta + 2\alpha$ colours in $O(m\log n)$ time. Moreover, the colouring returned by our algorithm has the following local property: every edge $uv$ is coloured with a colour in $\{1, \max\{\deg(u), \deg(v)\} + 2\alpha\}$. The time bound matches that of the greedy algorithm that computes a $2\Delta-1$ colouring of the graph's edges, and improves the number of colours when $\alpha$ is sufficiently small compared to $\Delta$.

\end{abstract}
\newpage 

\section{Introduction and related work}
When working on rapidly evolving, large scale graphs, algorithms need to adapt to the change in data quickly.
The dynamic model is interested in maintaining some property in a graph undergoing edge insertions and/or deletions, and has led to many fast algorithms, with polylogarithmic update and query time, in particular through the use of amortized analysis.

Graph colouring is a family of fundamental problems with many applications in computer science. We study the edge-colouring problem: the goal is to assign edges colours such that edges sharing an endpoint are coloured differently.
This problem has implications in resource allocation and scheduling, for example to allocate bandwidth in an optical network~\cite{ERLEBACH200133}.

A \emph{$C$-edge colouring} of a graph $G=(V, E)$ can be represented by a function $f:E\rightarrow\{1,...,C\}$, and the smallest palette size $C$ for which there exists a proper $C$ edge colouring is called the edge chromatic number of $G$, denoted $\chi'$.
If $\Delta$ is the maximum degree of $G$, then the edge chromatic number is clearly at least $\Delta$.
Vizing~\cite{Vizing1965} proved that $G$ can always be coloured with $\Delta + 1$ colours. On the other hand, Holyer showed that it is NP-complete to determine the edge chromatic index of an arbitrary graph~\cite{Holyer}, and the problem remains NP-complete even for cubic graphs. 
A $(\Delta(uv) + C)$-edge colouring is a proper colouring of the graph where each edge $uv$ receives a colour from $\{1,..., \Delta(uv) + C\}$. Here $\Delta(uv) = \max\{d(u),d(v)\}$.

Vizing's proof is constructive and suggests a way to extend a proper partial colouring to a larger subgraph by recolouring $O(\Delta +n)$ edges. Furthermore, the colour changes can be performed in polynomial time. 
However, so far the fastest algorithms for statically $(\Delta + 1)$ edge colouring a graph spend $O(m \sqrt{n})$~\cite{Sinnamon} or $\tilde{O}(m \Delta)$~\cite{Gabow} time. For certain graphs, faster algorithms are known~\cite{DBLP:journals/corr/abs-2303-05408,Bhatta23}.

It is interesting to see whether one can reduce the running time by slightly increasing the palette size. This line of research has been pursued before. In particular, the problem of $2\Delta - 1$ edge colouring can be solved greedily, yielding static algorithms running in near-linear time~\cite{Cole2001} and dynamic algorithms with $O(\log\Delta)$ update time~\cite{BCHN18}. In this dynamic setting, there are known algorithms~\cite{christiansen2023power,DHZ19} that achieve a randomized $(1 + \epsilon)\Delta$ colouring in $\operatorname{poly}(\log n, \epsilon^{-1})$ time, with $\epsilon > 0$ by Duan, He, and Zhang~\cite{DHZ19}, and later Christiansen~\cite{christiansen2023power}.

In the distributed setting, Chang, He, Li, Pettie and Uitto~\cite{CHLPU18} designed a randomized $\Delta + O(\sqrt{\Delta})$ edge colouring in $\operatorname{poly}(\log n)$ rounds based on the Lovasz Local Lemma and Barenboim, Elkin and Maimon~\cite{BEM16} describe a simple deterministic distributed algorithm for $\Delta + O(\alpha)$ colouring in polylogarithmic time. 
There are many more papers achieving different trade-offs between time and palette size. See for instance~\cite{BERNSHTEYN, Davies23, ghaffari2018deterministic, su2019towards} for different trade-offs in the distributed setting, or the papers~\cite{CHLPU18, ghaffari2018deterministic} for a more extensive discussion. 

For some algorithmic problems, especially ones where recourse is an important part of the running time of an algorithm, or the recourse is of interest on its own, the best known analysis follows a specific proof strategy; \emph{“solution oblivious analysis”}. By solution oblivious, we mean that we do not only give guarantees against worst case input graphs at each step of the algorithm, we furthermore always, when analysing the next step, are robust against an adversary changing the solution to the worst-possible solution before every update. 
Examples of such analysis include 
the analysis of the SAP protocol for maintaining maximum matchings in bipartite graphs of Bernstein, Holm and Rotenberg~\cite{BHR19}, 
the analysis of recourse in the edge-colouring algorithms of Bernshteyn~\cite{BERNSHTEYN}, Christiansen~\cite{christiansen2023power}, and Duan, He and Zhang~\cite{DHZ19}, 
and the analysis of the fully-dynamic out-orientation scheme due to Brodal \& Fagerberg~\cite{WADS97}, in which a potential function bounds the reorientations by comparing to an \emph{existing} (but not necessarily efficient) algorithm that augments paths whose lengths are bounded in an oblivious manner.

For the edge-colouring problem, an interesting lower bound has been proved by Chang, He, Li, Pettie and Uitto: they show that there exists a graph, and a partial colouring of this graph with a single uncoloured edge, such that to colour this edge, one needs to recolour a subgraph of diameter $\Omega(\frac{\Delta}{c} \log(\frac{cn}{\Delta} ))$. This is then also a lower bound on the number of edges that need to be recoloured.
This means that if we restrict ourselves to solution oblivious analysis, a dynamic algorithm with polylogarithmic update time will need a palette of size at least $\Delta + O\left(\frac{\Delta}{poly(\log n)}\right)$.
The analysis of Christiansen~\cite{christiansen2023power}, and Duan, He and Zhang~\cite{DHZ19} for their $(1+\epsilon)\Delta$ dynamic edge-colouring algorithms are solution oblivious, which results in algorithms that have a polynomial dependency on $\epsilon^{-1}$.

Whether one can design an algorithm with $\text{poly}\log$ update time that only uses $\Delta + O(\Delta^{1-\varepsilon})$ colours for some constant $\varepsilon >0$ remains a fundamental open problem. Improved results for special classes of graphs, like forests, planar graphs or sparse graphs, also receive attention in the community. In the static setting, it was shown by Bhattacharya, Costa, Panski and Solomon~\cite{Bhatta23} that one can compute a $(\Delta+1)$ edge colouring in $\tilde{O}(\min \{m \sqrt{n}, m \Delta \} \cdot{} \frac{\alpha}{\Delta})$-time. Here $\alpha$ is the \emph{arboricity} of the graph, and it is equal to the smallest number of forests needed to cover the edges of a graph. It is within a factor of 2 of other sparsity measures like the maximum subgraph density and the degeneracy.  Many other problems like, for instance, maintaining dynamic matchings~\cite{DBLP:conf/icalp/KopelowitzKPS14,DBLP:conf/soda/PelegS16}, maintaining a dynamic data structure that can answer adjacency queries efficiently~\cite{WADS97} and maintaining a maximal independent set~\cite{DBLP:conf/icalp/OnakSSW18} also have solutions that run faster in graphs with low arboricity.

\subparagraph{Our contribution.}
A natural question is therefore to ask if one can further reduce the palette size in dynamic graphs that are at all times sparse. In this paper, we show that this is the case. More specifically, we show that there exists a dynamic algorithm that can maintain an edge colouring with only $\Delta + O(\alpha)$ colours in poly-logarithmic update time.
Since the arboricity can be as large as $\frac{\Delta}{2}$, this is not always an improvement, however for many graph classes like forests, planar graphs, and graphs with constant arboricity, the number of colours used is significantly reduced compared to other efficient dynamic edge-colouring algorithms.

Tables \ref{tab:soa_static}, \ref{tab:soa} and \ref{tab:soa_distributed} summarise the results mentioned above and are not a comprehensive overview of the state of the art.

\begin{table}[ht]
    \centering
    \def\arraystretch{2}
    \begin{tabular}{|l|c|l r|}
        \hline
        Palette size & Time & Notes & Reference
        \\ \hline
        $\Delta$ & $O(m\log\Delta)$ & bipartite multigraph & \cite{Cole2001}
        \\ \hline
        $2\Delta(uv) - 1$ & $O(m\log\Delta)$ &  & \cite{BCHN18}
        \\ \hline
        $\Delta + 1$ & $O(m\sqrt{n})$ & randomised & \cite{Sinnamon} 
        \\ \hline
        $\Delta + 1$ & $\tilde O(m\Delta)$ &  & \cite{Gabow} 
        \\ \hline
        $\Delta + 1$ & $\tilde{O}(\min \{m \sqrt{n}, m \Delta \} \cdot{} \frac{\alpha}{\Delta})$ &  & \cite{Bhatta23}
        \\ \hline
        $\Delta(uv) + 1$ & $O(n^2\Delta)$ &  & \cite{christiansen2023power} 
        \\ \hline
        $\Delta(uv) + 2\alpha - 2$ & $O(m\log\Delta)$ &  & new
        \\ \hline
    \end{tabular}
    \vspace{.5em}
    \caption{A comparison of static edge-colouring algorithms}
    \label{tab:soa_static}
\end{table}

\begin{table}[ht]
    \centering
    \def\arraystretch{2}
    \begin{tabular}{|c|c|l r|}
        \hline
        Palette size & Update time& Notes & Reference
        \\ \hline
        $2\Delta(uv) - 1$ & $O(\log\Delta)$ & worst case &\cite{BCHN18}
        \\ \hline
        $(1+\epsilon)\Delta $& $O(\log^9n\log^6 \Delta/\epsilon^6)$ & worst-case, randomised & \cite{christiansen2023power}
        \\ \hline
        $(1+\epsilon)\Delta$ & $O(\log^8 n/\epsilon^4)$ & amortized, randomised  & \cite{DHZ19}\\
        && $\Delta \in \Omega(\log^2 n/\epsilon^2)$ &
         \\ \hline
        $\Delta_{max} + O(\alpha_{max})$ & $O(\log n \log\Delta_{max})$
         & amortized & new\\
        $\Delta(uv) + O(\alpha)$ & $O(\log^2\! n \log\alpha_{max}\log\alpha\log\Delta_{max})$
         & amortized & new
        \\ \hline
    \end{tabular}
    \vspace{.5em}
    \caption{A comparison of dynamic edge-colouring algorithms.
    If $G$ goes through a sequence of insertions and deletions $G_1...G_T$, $\Delta_{max} = \underset{1\leq t \leq T}{max} \Delta(G_t)$ is the maximum $\Delta$ on all graphs in the sequence. $\alpha_{max}$ is defined similarly.
    }
    \label{tab:soa}
\end{table}

\begin{table}[ht]
    \centering
    \def\arraystretch{2}
    \begin{tabular}{|c|c|l r|}
        \hline
        Palette size &
        Rounds & Notes & Reference
         \\ \hline
        $\Delta + O(\alpha)$ & $O(\sqrt{\alpha} \log n)$ & LOCAL model & \cite{BEM16}
        \\ \hline
        $\Delta + 1$ & $poly(\Delta,\log n)$ & LOCAL model
        & \cite{BERNSHTEYN}
        \\ \hline
    \end{tabular}
    \vspace{.5em}
    \caption{State-of-the-art for edge-colouring algorithms in the LOCAL model.}
    \label{tab:soa_distributed}
\end{table}

\subparagraph{Independent work} In independent and concurrent work, Bhattacharya, Costa, Panski, and Solomon also maintain a $\Delta+O(\alpha)$ edge colouring in amortised polylogarithmic time per insertion or deletion~\cite{BCPS}.

\subsection{Notations}
In this paper, we focus on simple graphs. According to many definitions of edge colouring, an edge from a vertex to itself can not be coloured.

Let $G=(E, V)$ be the undirected input graph and $H$ a subgraph of $G$.
Define $\Gamma_H(v)$ to be the neighbourhood of $v$ with respect to a graph $H$, and given a subset of vertices $U\subseteq V$, define $\Gamma_U(v)$ to be the neighbourhood of $v$ with respect to the subgraph induced by $U$ in $G$. Later, we use $N_H(v)$ to refer to a data structure containing $\Gamma_H(v)$. Define $\deg_H(v)$ to be the degree of $v$ with respect to a graph $H$. Formally, given $H\subseteq G$ and $U\subseteq V$:
\begin{align*}
    \Gamma_H(v) &= \{u\in V\ |\ (u, v)\in E(H)\}\\
    \Gamma_U(v) &= \{u\in U\ |\ (u, v)\in E(G)\}\\
    \deg_H(v) &= |\Gamma_H(v)|\\
    \deg_U(v) &= |\Gamma_U(v)|\\
\end{align*}
Let $G$ go through a sequence of insertions and deletions $G_1...G_T$.
At any iteration $t$, we denote by $\Delta_t(uv)$ the maximum degree of the endpoints of the edge $uv$ and $\Delta_t$ the maximum degree of the graph considered. $\Delta_{max}$ is the maximum $\Delta$ on all graphs in the sequence. Formally, at iteration $t$:

\begin{align*}
    \Delta_t(uv) &= \max\{\deg_{G_t}(u), \deg_{G_t}(v)\}\\
    \Delta_t &= \Delta(G_t)\\
    \Delta_{max} &= \underset{1\leq t \leq T}{max} \Delta(G_t)
\end{align*}

If the context is not ambiguous, we drop the subscript and write $\Delta, \Delta(uv)$ to refer to the maximum degree and the maximum degree between $u$ and $v$ in the current iteration.

\subparagraph{Arboricity}
The arboricity $\alpha$ of a graph $G=(V, E)$ is defined as:
\[\alpha = \underset{U\subseteq V, \ |U|>1}{max}\left\lceil \frac{|E(U)|}{|U|-1}\right\rceil\]
On a more intuitive level, the arboricity can also be defined as the smallest number $\alpha$ such that the edges of the graph can be partitioned in $\alpha$ forests. The two definitions are equivalent by Nash-Williams theorem~\cite{NWthm}.
A relevant consequence is that there exists an orientation of $G$ where each vertex has at most $\alpha$ out-neighbours.

\subparagraph{H-partition}
Let $\mathit{H} = \{H_1, ... H_k\}$ be a partition of the vertex set $V(G)$.
If a vertex $v$ is in $H_i$, we say that the level of $v$ is $l(v) = i$.
We denote by $Z_i = \bigcup_{j\geq i} H_j$ the vertices in levels $i$ and above (Figure~\ref{fig:HP_notations}). We may abuse these notations and use $H_i$, $Z_i$ to refer to the subgraph induced in $G$ by those sets.

\begin{figure}[ht]
\centering

\begin{tikzpicture}
  [level/.style={circle,draw=black!50,fill=white!20,thick,
                 minimum size=10mm},
    txt/.style={circle,draw=white!50,fill=white!20,thick,
                 inner sep=0pt,}
  ]
  \node[level]  (H1)    {$H_1$}
  ;
  \node[level]  (H2)    [right=of H1]   {$H_2$}
  ;
  \node[txt]   (etc)   [right=of H2]   {...}
  ;
  \node[level]  (Hk)    [right=of etc]   {$H_k$}
  ;
  
\draw [pen colour={black},
    decorate, 
    decoration = {calligraphic brace,mirror,
        raise=-5pt,
        amplitude=5pt,
        aspect=0.5}] (4.8,-1) --  (6,-1)
node[pos=0.25,below=5pt,black]{$Z_k = H_k$};

  \draw [pen colour={black},
    decorate, 
    decoration = {calligraphic brace,mirror,
        raise=-5pt,
        amplitude=5pt,
        aspect=0.5}] (1.5,-2) --  (6,-2)
node[pos=0.5,below=5pt,black]{$Z_2$};

  \draw [pen colour={black},
    decorate, 
    decoration = {calligraphic brace,mirror,
        raise=-5pt,
        amplitude=5pt,
        aspect=0.5}] (-0.5,-3) --  (6,-3)
node[pos=0.5,below=5pt,black]{$Z_1 = V(G)$};
\end{tikzpicture}
    
\caption{Hierarchical partition of the vertex set of a graph $G$.}
\label{fig:HP_notations}
\end{figure}
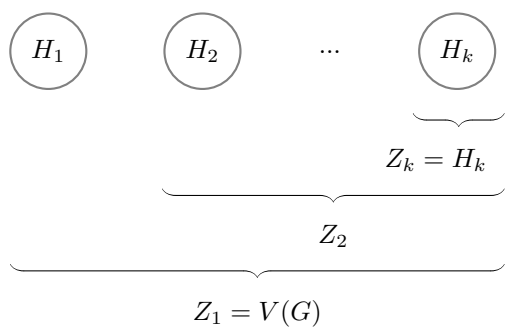

\subparagraph{Orientation} Consider the following orientation of the graph: if $l(u) < l(v)$, we orient the edge from $u$ to $v$ (Figure~\ref{fig:HP}). If $u$ and $v$ are on the same level, we have an edge in both directions. This orientation enables us, given a vertex $v$, to refer to the neighbours of $v$ in $Z_{l(v)}$ as the out-neighbours of $v$. We denote by $\deg^+(v)$ the out-degree of~$v$.
\[\deg^+(v) = \deg_{Z_{l(v)}}(v)\]

\begin{figure}[ht]
\centering

\begin{tikzpicture}
    [level/.style={circle,draw=black!50,fill=white!20,thick,
                   minimum size=10mm},
      txt/.style={circle,draw=white!50,fill=white!20,thick,
                   inner sep=0pt,}
    ]
    \node[level]  (H1)    {$H_1$}
    ;
    \node[level]  (H2)    [right=of H1]   {$H_2$}
        edge [pre]            (H1)
    ;
    \node[txt]   (etc)   [right=of H2]   {...}
      edge [pre, bend right]  (H2)
      edge [pre, bend left] (H1)
    ;
    \node[level]  (Hk)    [right=of etc]   {$H_k$}
        edge [pre, bend left]  (H2)
        edge [pre, bend right]  (H1)
    ;
\end{tikzpicture}
    
\caption{Hierarchical partition. In the static setting, the out degree of a vertex is bounded by $d$.}
\label{fig:HP}
\end{figure}
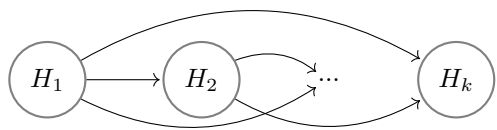

\subsection{Palettes}
In the following, a palette is a data structure that keeps track of the colours that are used or not at a vertex. They cover colours $[1, 2^{\lceil\log(2\Delta-1)\rceil}]$, where the value of $\Delta$ can change.
In the dynamic algorithms, we will use the following result:
\begin{theorem}[Palettes]
\label{theorem:palettes}
    Consider two palettes $P$, $Q$ such that there are $a$ colours used in $P$ and $b$ colours used in $Q$. We can find a colour that is available in $[1, a + b + 1]$ in $\log \Delta$ time.
\end{theorem}
\begin{proof}
    See lemmas \ref{lemma:time_palette} and \ref{lemma:n_colours}.
\end{proof}

\subsection{Roadmap}
Baremboim, Elkin and Maimon describe a simple distributed algorithm for $\Delta + O(\alpha)$ colouring~\cite{BEM16}. They form a hierarchical partition (or H-partition) of the graph, colour each set greedily, then colour the edges going out of the sets in an appropriate order. In Section~\ref{section:static}, we show that this technique can easily be adapted into an efficient algorithm in the static setting. We also show a simpler algorithm that yields a $\Delta + 2\alpha - 2$ edge colouring within the same time by building a degeneracy order of the graph instead of an H-partition.

In Sections~\ref{section:dynamic_max} and~\ref{section:dynamic_full}, we maintain a valid edge colouring in poly logarithmic time. We first present a simplified version of our algorithm that maintain a fully dynamic $\Delta_{max} + O(\alpha_{max})$ edge colouring in Section~\ref{section:dynamic_max}. This relies on two ideas: first, we maintain a dynamic H-partition, which only requires simple changes from the decomposition of Bhattacharya, Henzinger, Nanongkai, and Tsourakakis~\cite{BHNT15}, namely, we need to maintain two palettes of available colours at each vertex, one for its neighbours and one for its out-neighbours. Then it is easy to colour an edge in a valid partition.

Then we present an algorithm that is adaptative to the maximum degree and arboricity and maintains a $\Delta(uv) + O(\alpha)$ edge colouring. Our data structure is derived from the Level Data Structure of Henzinger, Neumann and Wiese~\cite{HNW20}, which has the following property: instead of having an out-degree that depends on $\alpha$, the levels of the partition have an increasingly large out-degree, and the sequence of maximum out-degree is such that it will be bounded by the current value of the arboricity, within a constant factor. Adapting to the current maximum can be done by updating the few problematic neighbours of a vertex when its degree decreases.

In Appendix~\ref{section:constants}, we study the constants, and show that with small modifications of the data structure, we can get $\Delta_{max} + (4+\epsilon)\alpha_{max}$ and $\Delta(uv) + (8+\epsilon)\alpha$ colours in the dynamic setting.

\section{Static \texorpdfstring{$\Delta(uv) + O(\alpha)$}\ \  colouring}
\label{section:static}
In this section, we describe a static edge-colouring algorithm. We arrange the vertices in a hierarchical partition that results in a $O(\alpha)$ out-orientation, then we colour vertices from right to left, so that for one of the endpoints, only the out-edges may already have colours. As the other endpoint has at most $\Delta$ coloured edges, the algorithm results in a $\Delta + O(\alpha)$ edge colouring.

The partition as it will be a crucial part of the dynamic algorithm. However, in the static setting, we could perform this algorithm with any $\alpha$ out-orientation. We show how this leads to a $\Delta + 2\alpha - 2$ edge colouring within the same asymptotic running time. We describe this algorithm first.

\subparagraph{Arboricity and degeneracy}

A graph $G$ of arboricity $\alpha$ can be partitioned into $\alpha$ forests has at most $\alpha(n-1)$ edges, therefore it has a vertex of degree at most $2m/n \leq 2\alpha - 1$: the degeneracy of $G$ is less than twice the arboricity. Therefore, the degeneracy is a 2-approximation of the arboricity. We can compute a degeneracy order of the graph, in linear time~\cite{MB83}, and colour the edges as follows: we colour the out-edges of the vertices from right to left. We describe this in more details in the following proof.

\begin{theorem}
Given a graph of arboricity $\alpha$, we can compute a $\Delta(uv) + 2\alpha - 2$ colouring of the graph in $O(m\log \Delta)$ time.

\end{theorem}
\begin{proof}
For this proof only, let $d \leq 2\alpha - 1$ denote the degeneracy of the graph. We compute a degeneracy ordering, which can be done in linear time~\cite{MB83}.
We greedily colour the edges from $v_i$ to $v_{j > i}$ for $i=n-1,...,1$. Let us prove by induction that we can always find an available colour in a palette of size $\Delta(uv) + d - 1$. At the first iteration, we can greedily colour a potential edge from $v_{n-1}$ to $v_n$ with one colour.

Assume that for $i<n-1$, we have coloured any edge $v_jv_{j'}$ s.t. $j, j' > i$ with at most $\Delta(v_jv_{j'}) + d - 1$ colours. Consider an uncoloured edge $e=v_iv_j$ with $j > i$.
$v_j$ is incident to at most $\deg(v_j) - 1$ coloured edges and $v_i$ is adjacent to at most $d-1$ coloured edges, therefore $e$ sees at most $\deg(v_j) + d -2$ colours and can find an available colour in a palette of size $\deg(v_j) + d - 1$.

If we maintain a binary tree over the palette at each vertex, we can find an available colour for an edge in $O(\log \Delta)$ time according to Theorem~\ref{theorem:palettes}. Therefore, colouring all the edges takes $O(m\log \Delta)$ time.
\end{proof}
This could conclude the static version. However, this algorithm does not translate into an efficient dynamic algorithm, as far as we can say; therefore we introduce the notion of hierarchical partition, which will be crucial in our dynamic algorithms. The rest of the section proves a slightly weaker theorem through the use of such a partition:
\begin{theorem}
    We can compute a static $\Delta(uv) + O(\alpha)$ edge colouring in $O(mlog\Delta)$ time.
\end{theorem}

\subparagraph{H-partition}
In the following, we want to compute a partition such that in the corresponding orientation, each vertex has out-degree at most $O(\alpha)$. Barenboim, Elkin and Maimon H-partition from~\cite{BE10} describe a distributed algorithm which translates well to the static setting:
we call a vertex \emph{active} if it does not have a level assigned yet. Initially, all the vertices are active. A vertex that is active at iteration $i$ will be part of the set $Z_i$. We call active degree the number of active neighbours of a vertex.
Initially, the active degree of a vertex is its degree. Then at iterations $i=1...k$, we group all the vertices of active degree at most $d = 4\alpha$ into a set $H_i$. For each vertex in $H_i$, we decrement the active degree of its neighbours and continue.

\begin{lemma}
\label{lemma:HP}
    We can compute a H-partition $\mathit{H} = \{H_1, ... H_k\}$ of a graph G of arboricity $\alpha$ such that the size of the partition is at most $k = \lfloor \log n\rfloor + 1$ and for all i, for all $v\in H_i$, the degree of $v$ in $G(Z_i)$ is at most $4\alpha$.
\end{lemma}
\begin{proof}
Consider a set $H_i$.
The average degree is:
\[2\frac{|E(Z_i)|}{|V(Z_i)|} \leq 2\alpha\leq \frac{4\alpha}{2}\]
At most half of the vertices have a degree more than $4\alpha$, otherwise we would have an average degree greater than $2\alpha$, leading to a contradiction. Therefore:
\begin{align*}
    |Z_{i+1}| &\leq|Z_i|/2\\
    |Z_{i}| &\leq n/2^{i-1}\\
    |Z_{\lceil \log n\rceil + 1}| &< 1
\end{align*}
Therefore, we can safely set $k = \lceil \log n\rceil$
\end{proof}

\begin{lemma}
\label{HP}
We can compute the H-partition described in lemma \ref{lemma:HP} in $O(m)$ time.
\end{lemma}
\begin{proof}
The sum of the degrees is $2m$, so we cannot decrement the active degrees more than $O(m)$ time in total, and decrementing a degree takes constant time. Then, at iteration $i$, we can search the vertices of low active degree in $O(n/2^i)$ time. Therefore, the running time to compute the H-partition is $O(m + n)$.
\end{proof}

\begin{lemma}
Given a H-partition, we can compute a $\Delta(uv) + O(\alpha)$ colouring of the graph in $O(m\log \Delta)$ time.
\end{lemma}
\begin{proof}
Let $d=4\alpha$.

We greedily colour the edges from $H_i$ to $H_{j \geq i}$ for $i=k,...,1$. Let us prove by induction that we can always find an available colour in a palette of size $\Delta(uv) + d - 1$. At the first iteration, we can greedily colour $H_k$ with $2d-1$ colours.

Assume that for $i<k$, we have coloured $G(H_{i+1}\cup...H_k)$ with at most $\Delta + d - 1$ colours. Consider an uncoloured edge $e=uv$ such that $u\in H_i, v\in H_{i}\cup...H_k$. $v$ is incident to at most $\deg(v) - 1$ coloured edges and $u$ is adjacent to at most $d-1$ coloured edges, therefore $e$ sees $\deg(v) + d -2$ colours and can find an available colour in a palette of size $\deg(v) + d - 1$.

If we maintain a binary tree over the palette at each vertex, we can find an available colour for an edge in $O(\log \Delta)$ time according to theorem \ref{theorem:palettes}. Therefore, colouring all the edges given the H-partition takes $O(m\log \Delta)$ time using theorem \ref{theorem:palettes}.
\end{proof}

\section{Dynamic \texorpdfstring{$\Delta_{max} + O(\alpha_{max})$}\ \   colouring}
\label{section:dynamic_max}
In this section, we show how to update a dynamic edge colouring. We first discuss how we can recolour an edge within a valid H-partition, then we show how we can maintain such a partition.
The algorithm requires $\alpha_{max}$ and $\Delta_{max}$ to be known in advance.

\subsection{Data structure}
We consider a H-partition that maintains the following invariants, with $d = 4\alpha_{max}$:
\begin{enumerate}
    \item Each vertex $v$ such that $l(v) > 1$ has at most $\beta d$ neighbours in $Z_{l(v)}$ (out-neighbours). In this section, we choose $\beta = 5$.
    \item Each vertex $v$ has at least $d$ neighbours in $Z_{l(v) - 1}$
\end{enumerate}
For each vertex $v \in H_i$, we store the following:
\begin{itemize}
    \item For each $j < i$, we store the neighbours of $v$ at level $j$, 
    $\Gamma_{H_j}(v)$,
    in a linked list $N_{H_j}(v)$.
    \item We store the out-neighbours of $v$, $\Gamma_{Z_i}(v)$, in a linked list $N_{Z_i}(v)$.

    \item We store the length of each linked list.
    \item For each edge $uv$, we store a pointer to the position of $u$ in the appropriate neighbour list
    $N_\cdot(v)$
    and conversely.
    \item We store two palettes: one for the neighbours of $v$ and one for its out-neighbours. We refer to those palettes as $P_{G}(v)$ and $P_{Z_i}(i)$.
\end{itemize}

\subsection{Recolouring an edge}
Let us start with a key sub problem: given an uncoloured edge in an otherwise valid data structure, what is the cost of colouring the edge? We colour the edge $uv$ as we would have in section \ref{section:static}.

Let $u$ be the leftmost vertex, i.e. $i = l(u) \leq l(v)$. We ignore the neighbours of $u$ that have a lower level, which can be done in practice by searching for an available colour in $P_{Z_i}(u)\cap P_G(v)$. 
The colour picked may conflict with a single edge from a neighbour of $u$ on a lower level. If that is the case, we recolour that edge in the same way.

\begin{lemma}
An uncoloured edge $uv$ in an otherwise valid data structure can be coloured in $O(\min(l(u), l(v))\cdot \log\Delta)$ time.
\end{lemma}
\begin{proof}
Assume wlog. $l(u) \leq l(v)$. If $uv$ is not coloured, at most $\beta d-1$ colours are represented at $u$ in the palette $P(Z_{l(u)})$, and at most $\deg(v) - 1$ colours are represented at $v$. Therefore, there exists an available colour in the palette $[\deg(v) + \beta d - 1]$.
If there exists $w$ such that $uw$ conflict with $uv$, then $l(w) < l(u)$. Therefore, the level of the left endpoint of the conflicting edge decrease at each iteration, which can happen at most $l(u) -1 $ times (Figure \ref{algo:recolour}).
Therefore, recolouring an edge and recursively resolving conflicts take $O(l(u)\cdot \log\Delta)$ time. 
\end{proof}

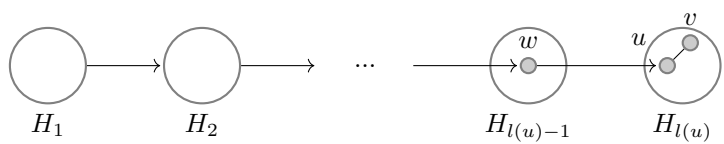
\begin{figure}[ht]
\centering

\begin{tikzpicture}
  [
    level/.style={
        circle,
        draw=black!50,
        thick,
        minimum size=10mm
    },
    vertex/.style={
        circle,
        draw=black!50,
        fill=black!20,
        thick,
        inner sep=0pt,
        minimum size=2mm
    },
    txt/.style={
        circle,
        draw=white!50,
        fill=white!20,
        thick,
        inner sep=10pt,
    },
  ]
  \node[level, label=below:$H_1$]  (H1)    {}
  ;
  \node[level, label=below:$H_2$]  (H2)    [right=of H1]   {}
      edge [pre]    (H1)
  ;
  \node[txt]   (etc)   [right=of H2]   {...}
      edge [pre]    (H2)
  ;
  \node[level, label=below:$H_{l(u) - 1}$]  (Him1)    [right=of etc]   {}
  ;
  \node[level, label=below:$H_{l(u)}$]  (Hi)    [right=of Him1]   {}
      ;
  \node[vertex, label=above:$w$]  (w)    [left=14mm of Hi]   {}
      edge [pre]  (etc)
  ;
  \node[vertex, label={[label distance=1mm]above left:$u$}]  (u)    [right=12mm of Him1]   {}
      edge [pre]  (w)
  ;
  \node[vertex, label=above:$v$]  (v)    [above right=2mm of u]   {}
      edge  (u)
      ;
\end{tikzpicture}
\caption{We may need to recolour at most $l(u)$ edges}
\label{fig:recolour}
\end{figure}

\begin{algorithm}[ht]
\caption{Recolour}
\label{algo:recolour}
\begin{algorithmic}

\Procedure{Recolour}{$uv$}
\If{$l(u) > l(v)$} 
\State $u, v = v, u$
\EndIf
\State $i = l(u)$
\State\emph{// Pick a colour available in $G(Z_i)$ in the palette $[\deg(v) + \beta d-1]$.}
\State $\Call{Colour}{uv, \deg(v) + \beta d - 1, P_{Z_i}(u), P_G(v)}$
\If{$c(uv)$ is represented at $u$}
    \State Use the pointer $C(u)[c]$ to find the conflicting edge $wu$
    \State $\Call{Recolour}{wu}$
\EndIf
\EndProcedure

\end{algorithmic}
\end{algorithm}

\subsection{Updating the hierarchical partition and full algorithm}
\subparagraph{Updates}
The algorithm for the updates is the following: let us say that a vertex that violates Invariant~1 or~2 is \emph{dirty}. As long as we have a dirty vertex $v$: if $v$ violate the first invariant, we increment its level. If it violates Invariant $2$, we decrement it. when doing so, we need to update our linked lists of neighbours and our tree palettes. Updating the trees is the limiting factor. The details of the updates are described in Algorithm \ref{algo:increment}. The algorithm terminates, which will be justified later, as we will define a positive potential that strictly decreases at each step. 

\begin{lemma}
    Increasing the level of a vertex takes $O\left(\deg^+(v)\log\Delta\right)$ time. Decreasing the level of a vertex takes $O\left(d\log\Delta\right)$ time.
\end{lemma}
\begin{proof}
When we increment the level of a vertex, we first update the lists of neighbours of $v$: we traverse $N_{Z_i}(v)$ to split it in $N_{H_i}(v)$ and $N_{Z_{i+1}}(v)$ in $O(\deg^+(v))$ time.
Then we discard $P_{Z_i}(v)$ and create $P_{Z_{i+1}}(v)$: we traverse the linked list $N_{Z_{i+1}}(v)$ and insert the $\deg_{Z_{i+1}}(v) = O(\deg^+(v))$ elements in the tree, which takes $O(\deg^+(v)\log\Delta)$ time. We also need to update the data structures of at most $\deg^+(v)$ neighbours, which takes constant time per neighbour for the linked lists and $\log\Delta$ time per neighbour for trees.

Finally, we need to check which vertices became dirty as a result of the operation. The set $Z_{i+1}$ has one more element, and for $j\neq i+1$, $Z_j$ is unchanged. Therefore, we check if $v$ itself, or any vertex in $N_{Z_{i+1}}(u)$, breaks the first invariant, which takes 
$O(\deg_{Z_{i + 1}}(v))$ time.

The procedure for decrementing a level is similar. If we decrement the level of a vertex, we know that the second invariant was not respected, i.e. $\deg_{Z_{i-1}}(v) < d$. After the operation, $Z_i$ does not include $v$ any more and for $j\neq i$, $Z_j$ is unchanged.
To update the lists of neighbours of $v$, we merge $N_{H_{i-1}(v)}$ and $N_{Z_i}(v)$, which takes $O(\deg_{Z_{i-1}}(v)) = O(d)$ time.
We create the palette of $v$ for the set $Z_{i-1}$ in $O(\deg_{Z_{i-1}}(v)) \log\Delta = O(d \log\Delta)$ time. Then we need to update the neighbour lists and the palettes of the neighbours of $u$ in $Z_i$, which also takes constant time per neighbour for the lists and $O(\log\Delta)$ time per neighbour for the palettes. Finally, we check if $v$ or any of its neighbours in $Z_i$ is dirty, which takes constant time per vertex.
\end{proof}

\begin{algorithm}[ht]
\caption{Insert the edge $uv$}
\label{algo:add}
\begin{algorithmic}

\Procedure{Add}{$uv$}
\If {$l(u) > l(v)$}
    \State \Call{Add}{$vu$}
    \State End procedure.
\EndIf
\State Add $v$ to the out-neighbours of $u$ and store the corresponding pointer.
\If {$l(u)$ < l(v)}
    \State Add $u$ to $N_{H_{l(u)}}(v)$, the neighbours of $v$ at level $l(u)$
    \State Store the corresponding pointer.
\Else
    \State Add $u$ to the out-neighbours of $v$ and store the corresponding pointer.
\EndIf
\State\emph{// The palettes will be updated when $uv$ gets a colour.}
\State Check if $u$ or $v$ became dirty.
\State \Call{Recover}{}
\State \Call{Recolour}{$uv$}
\EndProcedure

\end{algorithmic}
\end{algorithm}

\begin{algorithm}[ht]
\caption{Delete the edge $uv$}
\label{algo:delete}
\begin{algorithmic}

\Procedure{Delete}{$uv$}
\If {$l(u) > l(v)$}
    \State \Call{Delete}{$vu$}
    \State End procedure.
\EndIf
\State Remove the colour of $uv$ from the palettes of neighbours of $u$ and $v$.
\State Remove the colour of $uv$ from the palette of out-neighbours of $u$.
\If {$l(u) == l(v)$}
    \State Remove the colour of $uv$ from the palette of out-neighbours of $v$.
\EndIf
\State Remove $v$ from the neighbours of $u$ using the corresponding pointer.
\State Remove $u$ from the neighbours of $v$ using the corresponding pointer.
\State Check if $u$ or $v$ became dirty.
\State \Call{Recover}{}
\EndProcedure

\end{algorithmic}
\end{algorithm}

\begin{algorithm}[ht]
\caption{Recursively fix the invariants}
\label{algo:recover}
\begin{algorithmic}

\Procedure{Recover}{}
\While{there exists a dirty vertex $v$}
\State $i = l(v)$
\If {$\deg^+(v) > \beta d$} \Call{Increment}{v}
\ElsIf {$\deg_{Z_{i-1}}(v) < d$} \Call{Decrement}{v}
\EndIf
\EndWhile
\EndProcedure

\end{algorithmic}
\end{algorithm}

\begin{algorithm}[ht]
\caption{Incrementing / decrementing the level of a vertex}
\label{algo:increment}
\begin{algorithmic}

\Procedure{Increment}{$v$}
\State $i = l(v)$.
\State Split $N_{Z_i}(v)$ in $N_{H_i}(v)$ and $N_{Z_{i+1}}(v)$.
\State Create the palette of $v$ for $Z_{i+1}$.
\State Discard the palette of $v$ for $Z_i$.
\For{$u \in N_{Z_{i+1}}(u)$}
    \State Using the pointer in $uv$, remove $v$ from $N_{H_i}(u)$, add $v$ to $N_{Z_{i+1}}(u)$ or $N_{H_{i+1}}(u)$.
    \State Update the pointers accordingly.
    \If{$l(u) = i+1$}
    \State Update the palette of $u$ for the set $Z_{i+1}$: add colour $c(uv)$.
    \EndIf
\EndFor
\State Increment $l(v)$
\EndProcedure

\Procedure{Decrement}{$v$}
\State $i = l(v)$.
\State Merge  $N_{Z_i}(v)$ and $N_{H_{i-1}}(v)$ into $N_{Z_{i-1}}(v)$.
\State Create the palette of $v$ for $Z_{i-1}$.
\State Discard the palette of $v$ for $Z_i$.
\For{$u \in N_{Z_i}(v)$}
    \State Move $v$ from $N_{H_i}(u)$ or $N_{Z_i}(u)$ to $N_{H_{i-1}}(u)$.
    \State Update the pointers accordingly.
    \If{$l(u) = i$}
    \State Update the palette of $u$ for the sets $Z_{i}$: remove colour $c(uv)$.
    \EndIf
\EndFor
\State Decrement $l(v)$
\EndProcedure

\end{algorithmic}
\end{algorithm}

\begin{theorem}
We can maintain a dynamic $\Delta_{max} + O(\alpha_{max})$ edge colouring of a graph in $O(\log n \log\Delta_{max})$ amortized update time.
\end{theorem}
\begin{proof}
    We define the following potential.
    \begin{align*}
        B &= \log\Delta_{max} \sum_{v\in V} \phi(v) + \log\Delta_{max} \sum_{e\in E} \psi(e)\\
        \phi(v) &= \sum_{j=1}^{l(v) - 1} \max(0, \beta d - \deg_{Z_j}(v))\\
        \psi(u, v) &= 2(k-\min(l(u), l(v))) + \mathbf{1}_{l(u)=l(v)}
    \end{align*}

    When we insert an edge, we create a potential $\psi(u, v) \leq 2k$. The potential of the other edges do not change and the potential of a vertex can only decrease, therefore the potential $B$ increase by at most $2k \log\Delta_{max}$. When we delete an edge, $\phi(u)$ and $\phi(v)$ increase by at most $k$ each, $\psi(u,v)$ is deleted, and the other potentials are not affected, therefore the potential increase by at most $2k\log\Delta_{max}$.

    When a dirty vertex increment its level, the cost of the operation will be paid by the drop in potential from the edges. When a dirty vertex decrements its level, the cost of the operation is paid by the drop in potential from the vertex, despite the increase in potential from the edges. For completeness, we repeat the details of the analysis that follows closely that of~\cite{BHNT15}. Let $i$ be the level of $v$ before the operation.
    \subsubsection*{Incrementing the level of a dirty vertex.}
    \begin{itemize}
        \item
         When $l(v)$ increases, we get $l(v) = i + 1$, so we add $\max(0, \beta d - \deg_{Z_i}(v))$ to $\phi(v)$.
        As invariant 1 was violated, we must have had $\deg_{Z_i}(v) > \beta d$ and therefore $\max(0, \beta d - \deg_{Z_i}(v)) = 0$: the potential of $v$ is unchanged.\\
        The potential of the other vertices can not increase (though it might decrease for a neighbour of $v$).
        \item The potential of an edge may only change if one of the endpoints is $v$ and the other endpoint $u$ verifies $l(u) \geq i$. Therefore, there are exactly $\deg^+(v)$ edges whose potential drop by one or two.
    \end{itemize}
    The total drop in potential is at least:
    \[\deg^+(v)\log\Delta_{max}=\Omega(\deg^+(v)\log\Delta)\]
    \subsubsection*{Decrementing the level of a dirty vertex.}
    \begin{itemize}
        \item We must have $\deg_{Z_{i-1}}(v) < d \Rightarrow \max(0, \beta d - \deg_{Z_{i-1}}(v) \geq (\beta - 1)d = 4d$. Therefore, the $\phi(v)$ drops by at least $4d$.
        \item If $u$ is a neighbour of $v$, $\deg_{Z_j}(v)$ is decremented if $j=i$ and is unchanged otherwise. This affects the potential of $u$ if $l(u) > i$. Therefore, $\sum_{u\in \Gamma(v)}\phi(u)$ increase by at most $\deg_{Z_{i+1}}(v) < \deg_{Z_{i-1}}(v) < d$.
        \item The potential of an edge may only change if one of the endpoints is $v$ and the other endpoint $u$ verifies $l(u) \geq i$. Therefore, there are exactly $\deg^+(v)\leq \deg_{Z_{i-1}}(v) \leq d$ edges whose potential increase by one or two.
    \end{itemize}
    The total drop of potential is at least:
    \begin{align*}
        \log\Delta_{max}\left(4d - d - 2d\right)
        = d\log\Delta_{max}
        &= \Omega(d\log\Delta) \qedhere
    \end{align*}
\end{proof}

\section{Dynamic \texorpdfstring{$\Delta(uv) + O(\alpha)$}\ \   colouring}
The data structure from the previous section could maintain a dynamic $\Delta_{max} + O(\alpha_{max})$ colouring. We modify it further to create a data structure that adapts to the maximum degree and arboricity. We adapt the arboricity by making a structure that does not depend explicitly on $\alpha$: instead, we give the level increasingly large out-degrees and show how the out-degree ends up being bounded by the current value of the arboricity within a constant factor. To adapt to the maximum degree, we locally adapt to the degree of the vertices: when the degree of a vertex decrease, we can recolour its problematic edges in polylogarithmic time.
\label{section:dynamic_full}
\subsection{Data Structure}

In the following, $L = 1 + \lceil \log  n\rceil$ . We now consider a H-partitions with $k= L\cdot \lceil \log n \rceil$ levels. 
Let $k'={\max}_{v\in V}\ l(v)$ denote the maximum level of any vertex, that is, the highest non-empty level.

We partition the levels into $\lceil \log n \rceil$ groups of size $L$ (Figure \ref{fig:HP_2}).

Let $g(v)$ denote the group of a vertex and $d(v) = 2^{g(v)}$.
We will maintain the following:
\begin{enumerate}
    \item Each vertex $v$ has at most $2\beta d(v)$ neighbours in $Z_{l(v)}$. In this section, we choose $\beta = 5$, so any vertex has at most $10d(v)$ out-neighbours.
    \item Each vertex $v$ has at least $d(v)$ neighbours in $Z_{l(v) - 1}$
    \item The colour of an edge $uv$ is chosen from the first $\Delta(uv) + 2\beta d(v)$ colours of the palette.\\
    This condition enables the data structure to adapt to changes in the arboricity. In the following, we prove that for any vertex, $g(v) \leq \lceil \log (4\alpha) \rceil$, which will result in a $\Delta(uv) + O(\alpha)$ colouring.
\end{enumerate}

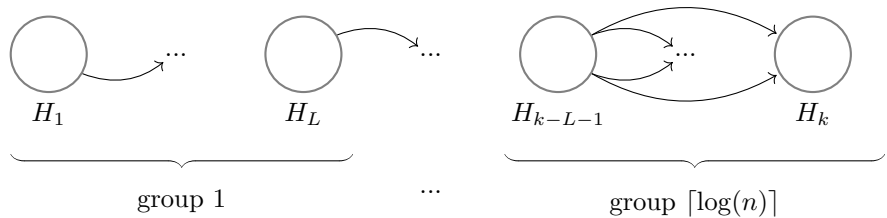
\begin{figure}[ht]
\centering

\begin{tikzpicture}
  [level/.style={circle,draw=black!50,fill=white!20,thick,
                 inner sep=0pt,
                 minimum size=10mm},
    txt/.style={circle,draw=white!50,fill=white!20,thick,
                 inner sep=0pt,}
  ]
  \node[level, label=below:$H_1$]  (H1)    {}
  ;
  \node[txt]   (etc1)   [right=of H1]   {...}
    edge[pre, bend left] (H1)
  ;
  \node[level, label=below:$H_L$]  (HL)    [right=of etc1]   {}
  ;
  \node[txt]   (etc2)   [right=of HL]   {...}
    edge[pre, bend right] (HL)
  ;
  \node[level, label=below:$H_{k-L-1}$]  (Hkm1)    [right=of etc2]   {}
  ;
  \node[txt]   (etc3)   [right=of Hkm1]   {...}
  edge[pre, bend left] (Hkm1)
  edge[pre, bend right] (Hkm1)
  ;
  \node[level, label=below:$H_k$]  (Hk)    [right=of etc3]   {}
  edge[pre, bend left] (Hkm1)
  edge[pre, bend right] (Hkm1)
  ;
    \draw [pen colour={black},
        decorate, 
        decoration = {calligraphic brace,mirror,
            raise=-5pt,
            amplitude=5pt,
            aspect=0.5}] (-0.5,-1.5) --  (4,-1.5)
    node[pos=0.5,below=5pt,black]{group $1$};

    \node[txt]  (etc4)    [below=15mm of etc2]   {...}
    ;

    \draw [pen colour={black},
        decorate, 
        decoration = {calligraphic brace,mirror,
            raise=-5pt,
            amplitude=5pt,
            aspect=0.5}] (6,-1.5) --  (11,-1.5)
    node[pos=0.5,below=5pt,black]{group $\lceil \log(n)\rceil$};
\end{tikzpicture}
    
\caption{Hierarchical partition with two levels. The in degree of a vertex $v$ is still only bounded by $\Delta$, when the bound on the out degree depends on its group $g(v)$.}
\label{fig:HP_2}
\end{figure}

We store the neighbours and palettes of the vertices as described in the previous section.

\begin{lemma}[Maximum level]
    The index of the highest non empty level, $k'$, is at most $O(\log \alpha \log  n)$.
\end{lemma}
\begin{proof}
Consider the group $\lceil \log (4\alpha) \rceil$ and a level $i$ in this group. We can repeat the arguments from the proof of lemma \ref{lemma:HP} and show that the last set of the group $Z_{L\lceil \log (4\alpha) \rceil}$, with $L= 1+\lceil \log n \rceil = O(\log n)$, has at most one element, therefore, all the higher groups are empty.
\end{proof}

\begin{lemma}
An uncoloured edge $uv$ in an otherwise valid data structure can be coloured in time $O(\log \alpha \log n \log\Delta)$.
\end{lemma}
\begin{proof}
Following the same reasoning as in the previous section, we have to recolour $O(k')$ edges, which we can do in $O(k'\log\Delta)$ time.
\end{proof}

\subsection{Updating the hierarchical partition and full algorithm}
When edges are deleted, the arboricity $\alpha$ or the maximum degree $\Delta$ may decrease. To adapt to the degree, we recolour the edges from in-neighbours of $v$ when $uv$ is deleted. To maintain the arboricity, we will recolour the edges to the out-neighbours of $v$ when its level decrease. As a result, we maintain that an edge $uv$, $l(u) \leq l(v)$, has a colour from the palette $[\deg(v) + \beta d(u) - 1]$

\begin{lemma}[Adapting to the maximum degree]
\label{lemma:delta}
    We can recolour the edges from the in-neighbours of $v$ in $O(\log n\log^2\alpha\log\Delta)$ time.
\end{lemma}
\begin{proof}
    In each group $i \leq g(v)$, $v$ might have at most one neighbour $u$ such that $uv$ has colour $\deg(v) + 2\beta d(u) - 1$. For each such group, it takes constant time to find this edge using the pointer $C(v)[\deg(v) + 10\cdot 2^i - 1]$. There are therefore at most $g(v)$ edges that we may have to recolour, each in $O(l(v)\log\Delta)$ time, when the degree of $v$ decreases. When we delete an edge, we do this for each of its two endpoints, which takes $O(g(v)l(v)\log\Delta)$ time.
\end{proof}

\begin{figure}[ht]
\centering

\begin{tikzpicture}
  [
    level/.style={
        circle,
        draw=black!50,
        thick,
        minimum size=10mm
    },
    vertex/.style={
        circle,
        draw=black!50,
        fill=black!20,
        thick,
        inner sep=0pt,
        minimum size=2mm
    },
    txt/.style={
        circle,
        draw=white!50,
        fill=white!20,
        thick,
        inner sep=10pt,
    },
  ]
  \node[level]  (H1)    {}
  ;
  \node[level]  (H2)    [right=of H1]   {}
  ;
  \node[txt]   (etc)   [right=of H2]   {...}
  ;
  \node[level]  (Him1)    [right=of etc]   {}
  ;
  \node[level]  (Hi)    [right=of Him1]   {}
      ;
  \node[vertex, label={below right:$u$}]  (u)    [right=12mm of Him1]   {}
  ;
  \node[vertex, label=above:$v$]  (v)    [above right=2mm of u]   {}
      edge [pre, bend right]  (H1)
      edge [pre, bend right]  (H2)
      edge [pre, bend right]  (Him1)
      edge  (u)
      ;

\end{tikzpicture}
    
\caption{When the degree of a vertex $v$ decrease, we may need to recolour at most $l(v)$ edges}
\label{fig:delta}
\end{figure}
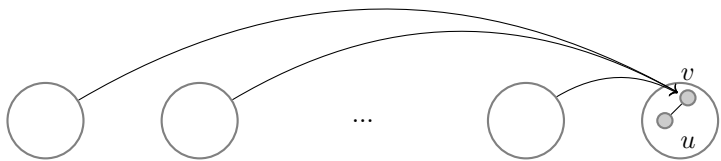

\begin{lemma}[Adapting to the arboricity]
    Increasing the level of a vertex $v$ takes $O(\deg^+(v)\log\Delta)$ time. Decreasing the level of a vertex takes $O(d\cdot k'\log\Delta)$ time.
\end{lemma}
\begin{proof}
The difference with the previous section is the following: when we decrement the level of a vertex $v$, we may need to recolour any edge $uv$ such that $\min(l(u), l(v))$ decrease, i.e. $u$ was on the same level as $v$ or higher. There are at most $\deg^+(v) = O(d)$ such edges, which can be recoloured in $O(k'\log\Delta)$ time.
\end{proof}

\begin{algorithm}[ht]
\caption{Decrementing the level of a vertex with adaptative arboricity}
\label{algo:increment_full}
\begin{algorithmic}

\Procedure{Adaptative Decrement}{$v$}
\State $i = l(v)$.
\State Merge  $N_{Z_i}(v)$ and $N_{H_{i-1}}(v)$ into $N_{Z_{i-1}}(v)$.
\State Create the palette of $v$ for $Z_{i-1}$.
\State Discard the palette of $v$ for $Z_i$.
\For{$u \in N_{Z_i}(v)$}
    \State Move $v$ from $N_{H_i}(u)$ or $N_{Z_i}(u)$ to $N_{H_{i-1}}(u)$.
    \State Update the pointers accordingly.
    \If{$l(u) = i$}
    \State Update the palette of $u$ for the sets $Z_{i}$: remove colour $c(uv)$.
    \EndIf
    \State \Call{Recolour}{$uv$}
\EndFor
\State Decrement $l(v)$
\EndProcedure

\end{algorithmic}
\end{algorithm}

\begin{theorem}
\label{theorem:adaptative}
We can maintain a dynamic $\Delta(uv) + O(\alpha)$ edge colouring of a graph in amortized $O(\log n\log\alpha_{max}\log\Delta_{max})$ time for insertions, $O(\log^2\! n \log\alpha_{max}\log\alpha\log\Delta_{max})$ for deletions.
\end{theorem}

\begin{proof}
    We define the following potential:
    \begin{align*}
        B &= k'_{max}\log\Delta_{max} \sum_{v\in V} \phi(v) + \log\Delta_{max} \sum_{e\in E} \psi(e)\\
        \phi(v) &= \sum_{j=1}^{l(v) - 1} max\left(0, \beta d(v) - \deg_{Z_j}(v)\right)\\
        \psi(u, v) &= 2(k'_{max}-\min(l(u), l(v))) + \mathbf{1}_{l(u)=l(v)}\\
        k'_{max} &= L \lceil\log(4\alpha_{max})\rceil \in O(\log n \log \alpha_{max})\\
    \end{align*}
    When we insert an edge, we create a potential $\psi(u, v) \leq 2k'_{max}$. The potential of the other edges do not change and the potential of a vertex can only decrease, therefore the potential $B$ increase by at most $\log\Delta_{max}\cdot 2k'_{max}$.
    When we delete an edge, we need $O(\log n\log^2\! \alpha\log\Delta)$ time to update the overflowing colours of the in-neighbours of $v$ (lemma \ref{lemma:delta}). Then for the potentials: $\phi(u)$ and $\phi(v)$ increase by at most $k'$ each, $\psi(u,v)$ is deleted, and the other potentials are not affected, therefore the potential increase by at most $k'_{max}\log\Delta_{max}\cdot 2k'$, therefore the costs.

    \subsubsection*{Incrementing the level of a dirty vertex.}
    Let $d$ denote the value of $d(v)$ before the change of level and $d'$ the value after the modification.
    \begin{itemize}
        \item If $v$ changes group, we have $d'=2d$, otherwise $d'=d$. Either way, we have $\deg^+(v) > 10d \Rightarrow \max(0, \beta d' - \deg^+(v)) = 0$.
        It follows that $\phi(v)$ is unchanged. The potential of the other vertices can not increase (though it might decrease for a neighbour of $v$).
        \item The potential of an edge may only change if one of the endpoints is $v$ and the other endpoint $u$ verifies $l(u) \geq i$. Therefore, there are exactly $\deg^+(v)$ edges whose potential drop by one or two.
    \end{itemize}
    The total drop in potential is at least:
    \[\deg^+(v)\log\Delta_{max} = \Theta(\deg^+(v)\log\Delta)\]
    \subsubsection*{Decrementing the level of a dirty vertex.}
    \begin{itemize}
        \item We must have $\deg_{Z_{i-1}}(v) < d \Rightarrow max\left(\beta d - \deg_{Z_{i-1}}(v)\right) \geq 4d$. Therefore, the $\phi(v)$ drops by at least $4d$. If the level of $v$ decreases, the potential only decreases further.
        \item If $u$ is a neighbour of $v$, $\deg_{Z_j}(v)$ is decremented if $j=i$ and is unchanged otherwise. This affects the potential of $u$ if $l(u) > i$. Therefore, $\sum_{u\in {\Gamma(v)}}\phi(u)$ increase by at most $6\deg_{Z_{i+1}}(v) < \deg_{Z_{i-1}}(v) < d$.
        \item The potential of an edge may only change if one of the endpoints is $v$ and the other endpoint $u$ verifies $l(u) \leq i$. Therefore, there are exactly $\deg^+(v)< \deg_{Z_{i-1}}(v) < d$ edges whose potential increase by one or two.
    \end{itemize}
    The total drop of potential is at least:
    \begin{align*}
        k'_{max}\log\Delta_{max}\left(4d-d\right) - \log\Delta_{max}\cdot 2d
        &\geq k'_{max} \log\Delta_{max}\cdot d \qedhere \\ 
    \end{align*}
\end{proof}

\section{Conclusion}

In this paper, we show how to maintain a $\Delta(uv) + O(\alpha)$ edge colouring in polylogarithmic time through the use of dynamic hierarchical partition. We also propose a simpler data structure to maintain a $\Delta_{max} + O(\alpha_{max})$ edge colouring, which can be done faster than the aforementioned algorithm. 

We give an amortized analysis of the running time of our dynamic algorithms. This raises the question of what can be done in worst case time. In our case, we are only limited by the updates of our hierarchical partitions, so it motivates the search for hierarchical partitions with efficient worst-case update times.

The question that motivated our research is still open for graphs that have a large arboricity compared to their maximum degrees: is it possible to maintain a $\Delta + O(\Delta^{1-\epsilon})$ edge colouring, with $\epsilon$ a positive constant, in polylogarithmic time?

In the static setting, we showed that we can make a $\Delta(uv) + 2\alpha - 2$ edge colouring in $O(m\log \Delta)$ time, which is as fast as the greedy $2\Delta(uv) - 1$ algorithm.
Thus, we get a $\Delta(uv) + O(1)$ edge colouring for graphs of constant arboricity, such as planar graphs, in near-linear time: more precisely, a planar graph has arboricity at most 3~\cite{Poh90}, so by our result, it can in near-linear time be edge-coloured with $\Delta(uv) + 2\alpha - 2 = \Delta(uv) + 4$ colours. 

Recently, it was shown by Bhattacharya, Costa, Panski and Solomon~\cite{Bhatta23} that one can compute a $(\Delta+1)$ edge colouring in $\tilde{O}(\min \{m \sqrt{n}, m \Delta \} \cdot{} \frac{\alpha}{\Delta})$-time, which gives a near linear time algorithm for graphs of polylogarithmic arboricity.
It emphasises the question whether a near-linear time $\Delta + O(1)$ edge-colouring algorithm could be obtained for a wider class of graphs. 

\newpage



\bibliography{references}

\appendix

\section{Constants}
\label{section:constants}
In this section, we assume that $\alpha$, or $\alpha_{max}$ if relevant, is known. 
In Sections~\ref{section:static}, section \ref{section:dynamic_max} and Section~\ref{section:dynamic_full}, we respectively compute or maintain a $\Delta(uv) + 4\alpha$, $\Delta(uv) + 20\alpha_{max}$, $\Delta(uv) + 40\alpha + O(1)$ edge colouring. In the static section, we also explained how to get a $\Delta(uv) + 2\alpha$. We show how much we can reduce the multiplicative constants in the dynamic settings.
\subsection{Reducing the outdegree}
\begin{lemma}
    In a graph $G=(V, E)$ of arboricity $\alpha$, the number of vertices of degree less than $d$ is at least $|V|\frac{2\alpha}{d}$.
\end{lemma}
\begin{proof}
Assume that more than $\frac{2\alpha}{d}|V|$ vertices have degree more than $d$. Then the sum of the degrees is:
\begin{align*}
    \sum_{v\in V}\deg(v) > \frac{2\alpha}{d} |V| d
    > 2|V| \alpha
    > 2|V| \frac{|E|}{|V|-1}
    > 2|E|
\end{align*}
Which leads to a contradiction.
\end{proof}

\begin{corollary}[static]
    Given $\epsilon > 0$ a constant, we can compute a H-partition $\mathit{H} = \{H_1, ... H_k\}$ of a graph G of arboricity $\alpha$ such that the size of the partition is at most $k = O(\frac{\log n}{\epsilon})$ and for all i, for all $v\in H_i$, the degree of $v$ in $G(Z_i)$ is at most $d=(2+\epsilon)\alpha$.
\end{corollary}
\begin{proof}
\begin{align*}
    |Z_{i+1}| &\leq\frac{2}{2+\epsilon}|Z_i|\\
    |Z_{i}| &\leq n\cdot\left(\frac{2}{2+\epsilon}\right)^{i-1}\\
    |Z_{\lfloor \log_{\frac{2+\epsilon}{2}} n\rfloor + 2}| &< 1
\end{align*}
Therefore, the number of sets is at most $\lfloor \log_{\frac{2+\epsilon}{2}} n\rfloor + 1 \leq \frac{2}{\epsilon} \log n = O(\frac{\log n}{\epsilon}
)$
\end{proof}

\subsection{Dynamic case}
\label{sub:constant_dyn_max}
\begin{theorem}
For any constant $\epsilon>0$, we can maintain a dynamic $\Delta_{max} + (4+\epsilon)\alpha_{max}$ edge colouring of a graph in $O(\frac{\log n}{\epsilon} \log\Delta_{max})$ amortized update time for insertions and $O(\frac{\log n}{\epsilon^2} \log\Delta_{max})$ for deletions.
\end{theorem}

\begin{proof}
    Let $\epsilon$ be a positive constant. We modify the data structure from Section~\ref{section:dynamic_max} as follows:
we choose $d=(2+\epsilon)\alpha => k = O(\frac{\log n}{\epsilon}$) and we use the following vertex potential:
\begin{align}
    \phi(v) &= \frac{1}{\epsilon}\sum_{j=1}^{l(v) - 1} \max(0, \beta d - \deg_{Z_j}(v))
\end{align}

When we insert an edge, we create a potential $\psi(u, v) \leq 2k$. The potential of the other edges do not change and the potential of a vertex can only decrease, therefore the potential $B$ increase by at most $2k \log\Delta_{max}$.

When we delete an edge, $\phi(u)$ and $\phi(v)$ increase by at most $k$ each, $\psi(u,v)$ is deleted, and the other potentials are not affected, therefore the potential increase by at most $\frac{2k}{\epsilon}\log\Delta_{max}$.

Now let us verify that we can pay for the operations. The edge potential is unchanged and the vertices potentials cannot increase when incrementing the level of a dirty vertex. What is left to review is the analysis for decrementing the level of a dirty vertex.

Let $i$ be the level of $v$ before the operation. When decrementing the level of a dirty vertex:
    \begin{itemize}
        \item We must have $\deg_{Z_{i-1}}(v) < d \Rightarrow \max(0, \beta/2\cdot d - \deg_{Z_{i-1}}(v) \geq (\beta - 1)d$. Therefore, the $\phi(v)$ drops by at least $\frac{d}{\epsilon}(\beta - 1)$.
        \item If $u$ is a neighbour of $v$, $\deg_{Z_j}(v)$ is decremented if $j=i$ and is unchanged otherwise. This affects the potential of $u$ if $l(u) > i$. Therefore, $\sum_{u\in N(v)}\phi(u)$ increase by at most $\frac{1}{\epsilon}\deg_{Z_{i+1}}(v) < \frac{1}{\epsilon}\deg_{Z_{i-1}}(v) < \frac{d}{\epsilon}$.
        \item As in Section~\ref{section:dynamic_max}, there are exactly $\deg^+(v)< \deg_{Z_{i-1}}(v) < d$ edges whose potential increase by one or two.
    \end{itemize}
    Therefore, the total drop of potential is at least:
    \begin{align*}
        \log\Delta_{max}\left(\frac{1}{\epsilon}(\beta d - 2d) - 2d \right)
        &= \log\Delta_{max}\cdot \frac{d}{\epsilon} (\beta - 2 - 2\epsilon)
    \end{align*}
    If we choose $\beta = 2 + 3\epsilon$, this equals $\log\Delta_{max}\cdot d$, which covers the cost of the operation, and the number of colours used is:
    \begin{align*}
        \Delta(uv) + \beta d &= \Delta(uv) + (2+3\epsilon)(2+\epsilon)\alpha\\
        &=\Delta(uv) + (4+\epsilon')\alpha\quad \text{with $\epsilon'$ a positive constant}
    \end{align*}
\end{proof}

\subsection{Dynamic case, adaptative}
\begin{theorem}
    For any constant $\epsilon>0$, we can maintain a dynamic $\Delta(uv) + (8+\epsilon)\alpha_{max}$ edge colouring of a graph in $O(\log n \log \alpha_{max}\log\Delta_{max}\epsilon^{-1})$ amortized update time for insertions, $O(\log^2 n \log \alpha_{max}\log\alpha \log\Delta_{max}\epsilon^{-3})$ for deletions.
\end{theorem}

\begin{proof}
Let $\epsilon$ be a positive constant.
We modify the data structure from Section~\ref{section:dynamic_full} as follows: we now take groups of size $L = 1 + \frac{2}{\epsilon}\log n \geq 1 + \log_{1+\epsilon}n$. We can verify that the group $\log(\lceil (2+\epsilon)\alpha \rceil)$ is the last non-empty group and therefore, for any vertex $v$, $d(v)\leq \lceil (2+\epsilon)\alpha \rceil$.

We use the following vertex potential:
\begin{align*}
    \phi(v) &= \frac{1}{\epsilon}\sum_{j=1}^{l(v) - 1} \max(0, \beta d - \deg_{Z_j}(v))
\end{align*}

The amortized costs of insertion and deletion are respectively 
$2k'_{max}\log\Delta_{max}$
and 
$2k'_{max}\log\Delta_{max}\cdot \frac{2 k'}{\epsilon}$
with 
$k'_{max} = L\lceil (2+\epsilon)\alpha \rceil \in O(\frac{\log n\log\alpha_{max}}{\epsilon})$, 
$k' \in O(\frac{\log n\log\alpha}{\epsilon})$.

As is subsection \ref{sub:constant_dyn_max}, the analysis is only affected when decrementing the level of a node. The potential of the nodes and vertices behaves as in subsection \ref{sub:constant_dyn_max} and the total drop of potential is at least:
    \begin{align*}
        k'_{max}\log\Delta_{max}\frac{1}{\epsilon}(\beta d - 2d) - \log\Delta_{max}2d 
        &\geq k'_{max}\log\Delta_{max}\cdot \frac{d}{\epsilon} (\beta - 2 - 2\epsilon)
    \end{align*}
If we choose $\beta = 2 + 3\epsilon$, this equals $k'_{max}\log\Delta_{max}\cdot d$, which covers the cost of the operation, and the number of colours used is:
    \begin{align*}
        \Delta(uv) + 2\beta d &= \Delta(uv) + 2(2+3\epsilon)
        \lceil (2+\epsilon)\alpha \rceil\\
        &=\Delta(uv) + (8+\epsilon')\alpha
    \end{align*}
Where $\epsilon'$ can be upper bounded by a function of $\alpha^{-1}$: \[\epsilon'\leq 22\epsilon + \frac{10}{\alpha}, \quad \epsilon \leq 1\]

\end{proof}

\section{Palettes}
\label{sec:palettes}
We describe the palettes from~\cite{BCHN18}. This will enable us to find a colour for an edge in $O(\log \Delta)$ time. We first describe the data structure assuming that $\Delta$ does not change.

Each vertex $v$ maintains the following:
\begin{itemize}
    \item An array $C(v)$ of length $2\Delta - 1$. Each entry corresponds to a colour. If the colour $c$ is represented at $v$, then $C(v)[c]$ is a pointer to the corresponding edge, otherwise it is a null pointer.
    \item A bit vector $A(v)$. $A(v)[c]$ is 1 if the colour is represented, 0 otherwise. The authors denote $A(v)[i:j]$ the number of colours represented at $v$ between $i$ included and $j$ excluded.
    \item A segment tree over $A(v)$. The leaves store the elements of $A(v)$. Each internal node store the sum of the leaves in its subtree.
\end{itemize}
For simplicity, we may assume the arrays and the tree of size $[1, 2^{\lceil\log(2\Delta-1)\rceil}]$, knowing that at most a constant fraction of the colours will not be used.
For convenience, we refer to the three components of the data structure as a palette $P(v)$. In the paper, we define the palette over subgraphs and given a subgraph $H$ we denote the palette of a vertex $v$ in $H$: $P_{H}(v)$. We also use the term palette as the universe of available colours.

\begin{algorithm}[h]
\caption{Algorithm Colour adapted from BCHN~\cite{BCHN18}}
\label{algo:colour}
\begin{algorithmic}
    \Procedure{COLOR}{Palettes of $u$ and $v$}
    \State $l\leftarrow 1, r \leftarrow [1, 2^{\lceil\log(2\Delta-1)\rceil}]$
    \While{$l<r-1$}
    \State $z \leftarrow\lceil (l + r)/2\rceil$
    \If{$A(u)[l:z] + A(v)[l:z] < z - l$}
    \State $r \leftarrow z$
    \ElsIf{$A(u)[z:r] + A(v)[z:r] < r-z$}
    \State $l \leftarrow z$
    \EndIf
    \EndWhile
    \State Assign colour $l$ to the edge $uv$
    \State Update the data structures $C(u)[l], C(v)[l], A(u)[l], A(v)[l], T(u), T(v)$
    \EndProcedure
\end{algorithmic}
\end{algorithm}

The algorithm \ref{algo:colour} is a binary search procedure. At each iteration, it chooses a half of the range where it is guaranteed that an available colour exists. We introduce the following difference with the original paper: the input consists of two palettes rather than two vertices. This change is significant because it allows us to purposefully ignore some edges by working with the palette of a vertex in a subgraph. Those edges are then recoloured in a second phase if conflict arises. Given two bit vectors $A(u)$ and $A(v)$, the colour returned is always one of the first $A(u)[1:2\Delta] + A(v)[1:2\Delta] + 1$ colours of the palette (see lemma \ref{lemma:n_colours}). In practice, given two palettes with respectively $A(u)[1:2\Delta] = \Delta - 1$ and $A(v)[1:2\Delta] = d - 1$ colours used, we can find a colour in $[\Delta + d - 1]$. The original paper follows from the fact that given an uncoloured edge $uv$, we necessarily have $A(u)[1:2\Delta], A(u)[1:2\Delta] \leq  \Delta - 1$

\begin{lemma}
\label{lemma:time_palette}
The running time of $\Call{Colour}{}$ is $O(\log\Delta)$.
\end{lemma}
\begin{proof}
    The number of iterations of the while loop is $O(\log \Delta)$. We only access sums that are stored in the internal leaves, so they can be accessed in constant time. Updating the arrays and bit vector takes constant time and the trees can be updated in $O(\log \Delta)$ time.
\end{proof}

\begin{lemma}
\label{lemma:n_colours}
Given two palettes $P(u)$ and $P(v)$, $\Call{Colour}$ (algorithm \ref{algo:colour}) return a colour in \[[A(u)[1:2\Delta] + A(v)[1:2\Delta] + 1]\]
\end{lemma}
\begin{proof}
First, we can prove that there is always a half of the interval that verifies the invariant. Initially, by definition of the sums, the number of colours unavailable is at most:
\[A(u)[1:2\Delta] + A(v)[1:2\Delta] + 1\]
Now assume that at a given iteration, the invariant is verified, i.e. $A(u)[l:r] + A(v)[l:r] < r - l$. If the invariant is not verified for either half of the interval, then we have:
\begin{align*}
    A(u)[l:r] + A(v)[l:r] &\geq r - l\\
    A(u)[z:r] + A(v)[z:r] &\geq r - z\\
    \Rightarrow A(u)[l:r] + A(v)[l:r] &\geq r - l
\end{align*}
Which contradicts our hypothesis. Therefore, the algorithm maintains the invariant successfully.

Now let us prove that we can find a colour in a palette of size $p=A(u)[1:2\Delta] + A(v)[1:2\Delta] + 1$. The algorithm maintains the invariant $A(u)[l:r] + A(v)[l:r] < r - l$ and recurse on half of the interval at each iteration, and it always favours the first interval if it verifies the invariant. Now assume that the algorithm returns a value greater than $p$. We can write the palette as a union of disjoint intervals that have been rejected, i.e. for each interval $I$, $A(u)[I] + A(v)[I] \geq |I|$. It follows that $A(u)[[p]] + A(v)[[p]] + 1 > p$, which contradicts our hypothesis.
\end{proof}

\paragraph*{Adaptative algorithm}
Let us now consider the case in which the maximum degree $\Delta$ can change. The arrays now need to be dynamic, so they still have size $[1, 2^{\lceil\log(2\Delta-1)\rceil}]$: we can access them in constant time and update in amortized constant time.
As insertions and deletions happen only on the right side of the tree, we can also make a dynamic binary tree as a dynamic array: we double the size of the tree when $2\Delta-1$ exceeds the number of leaves, and we delete the right half of the tree when less than the first quarter is used. The algorithm to find a colour is unchanged and updating the data structure can be done in amortized $O(\log\Delta)$ time.

In ~\cite{BCHN18}, the authors suggest using red black trees.

\end{document}